\documentclass[sigconf]{acmart}
\setcopyright{none}
\renewcommand\footnotetextcopyrightpermission[1]{} 
\usepackage{graphicx}
\usepackage{amsmath}
\usepackage{booktabs}
\usepackage{algorithm}
\usepackage{tikz}
\usepackage{amsfonts}
\usepackage{algorithmic}
\begin{document}
\title{Cost Sharing under Private Costs and Connection Control on Directed Acyclic Graphs}
%
%

\author{Tianyi Zhang, Dengji Zhao, Junyu Zhang, and Sizhe Gu}
\affiliation{%
  \institution{ShanghaiTech University}
  \city{Shanghai}
  \country{China}}
\email{{zhangty, zhaodj, zhangjy22022, guszh}@shanghaitech.edu.cn}





%
%
\begin{abstract}
We consider a cost sharing problem on a weighted directed acyclic graph (DAG) with a source node to which all the other nodes want to connect. The cost (weight) of each edge is private information reported by multiple contractors, and among them, only one contractor is selected as the builder. All the nodes except for the source need to share the total cost of the used edges. However, they may block others' connections to the source by strategically cutting their outgoing edges to reduce their cost share, which may increase the total cost of connectivity. To minimize the total cost of connectivity, we design a cost sharing mechanism to incentivize each node to offer all its outgoing edges and each contractor to report all the edges' weights truthfully, and show the properties of the proposed mechanism. In addition, our mechanism outperforms the two benchmark mechanisms.
\end{abstract}


\maketitle    

\section{Introduction}
\label{introo}
In the classic cost sharing problem, there is a group of agents at different locations and a source~\cite{bergantinos2010minimum,gomez2011merge,trudeau2017set}. All the agents want to connect with the source via the connections between the locations. Each connection has a public cost. The goal is to share the total cost of connectivity among all the agents~\cite{claus1973cost,granot1981minimum,lorenzo2009characterization}. This problem exists in many real-world applications such as electricity networks and communication networks, and many solutions have been proposed to satisfy different properties~\cite{bergantinos2007fair,hougaard2010decentralized,trudeau2012new,10.1007/978-3-662-54110-4_25}. 

In many practical scenarios, the connections are directed rather than undirected~\cite{feigenbaum2001sharing,penna2004sharing}. Meanwhile, there exist two natural strategic behaviors. First, to connect to the source, an agent may need to go through some intermediate agents. These agents will strategically block the connection if their cost share is reduced by doing so, which will potentially increase the total cost of connectivity. Second, building the connections between agents needs cost. Assume that multiple contractors compete to build the connections, and among them, only one contractor is selected as the builder. The costs are private information of contractors and need to be reported by them. However, they may misreport the costs for their interests. 

As an example, consider a water supply network. The source supplies water to villages at different positions possibly through other villages. The transmission of water from the source to a village or from one village to another needs a cost. Each cost is reported by multiple contractors. All the villages need to share the total cost of transmissions. To reduce the cost share, a village may block the transmissions from itself to other villages, which will potentially increase the total cost of transmissions. In addition, a contractor may misreport transmission costs to increase its utility. 

As another example, consider a multicast transmission network where the source is the owner of some data and each agent is a receiver of the data. The source transmits the data to the receivers. When a receiver receives the data, it duplicates and transmits the data to its adjacent receivers. Each cost of transmission is reported by multiple contractors and all the receivers need to share the total cost of the transmissions. However, a receiver may strategically block the transmissions of the data to reduce its cost share. Additionally, a contractor may misreport transmission costs to increase its utility.

On general networks, to minimize the total cost of connectivity, we design a cost sharing mechanism that can prevent the above two strategic behaviors. One challenge is how to deal with the conflict between the effort of mechanism designer to use all the connections to minimize the total cost and the motivation of agents to block the connections to reduce their cost share. This conflict is essentially the one between the system's optimality and the individuals' self-interests. Another challenge comes from the interaction between the behaviors of agents and contractors. Specifically, the decisions of agents to block connections affect the structure of the network, which further affect the reports of contractors. Conversely, the reports of contractors affect the decisions of agents to block connections.

To meet these challenges, we use a method inspired by the second price auction~\cite{vickrey1961counterspeculation,groves1973incentives} to select a contractor. Then according to the report of the selected contractor, we compute each agent's cost share based on its shortest path from the source to it. 

\emph{Our contribution:} We first model a typical scenario in real applications involving both agents and contractors that have strategic behaviors respectively. We then show that the cost sharing benchmark mechanisms based on the well-known Shapley value~\cite{shapley1953value} and Bird rule~\cite{bird1976cost} cannot satisfy truthfulness (i.e. each agent is incentivized to offer all its connections and each contractor is incentivized to truthfully report the costs of connections). Next, we design a novel truthful cost sharing mechanism. Finally, we show that the proposed mechanism also satisfies other desirable properties studied in the literature~\cite{gomez2017monotonic,norde2019degree,todo2020split}, i.e. budget balance, positiveness, individual rationality, ranking and symmetry.

\section{Related Work}
\label{relat}
There exists a rich literature on the cost sharing problem and they all modeled the problem as an undirected graph. We survey them from the following perspectives.

Some studies treated the problem from the perspective of non-cooperative game. Berganti{\~n}os and Lorenzo~\cite{bergantinos2004non} studied the Nash equilibrium of the problem. Tijs and Driessen~\cite{tijs1986game} proposed the cost gap allocation (CGA) method based on marginal contributions of agents, but it only applies to complete graphs. Bird~\cite{bird1976cost}, Dutta and Kar~\cite{dutta2004cost}, Norde \textit{et al.}~\cite{norde2004minimum}, Tijs \textit{et al.}~\cite{tijs2006obligation} and Hougaard \textit{et al.}~\cite{hougaard2010decentralized} provided cost sharing methods based on the minimum spanning tree (MST)~\cite{prim1957shortest} of a graph. Recently, Zhang \textit{et al.}~\cite{10.5555/3545946.3598890} considered a cost sharing mechanism design problem similar to ours. However, they assumed that the connection costs are public information. In contrast, we consider a more realistic scenario where the costs of connections are private information and need to be reported by multiple contractors. Therefore, our definition of truthfulness differs from theirs and their solutions cannot solve our problem.

Other studies treated the problem from the perspective of cooperative game. They are all based on the Shapley value and differ in the definition of the value of each coalition. Kar~\cite{kar2002axiomatization} proposed the Kar solution and defined the value of a coalition as the minimal cost of connecting all agents of coalition to the source without going through the agents outside of coalition. Bergantiños and Vidal-Puga~\cite{bergantinos2007fair} proposed the folk solution. They first computed the irreducible cost matrix, and then based on it, they defined the value of a coalition in the same way as the Kar solution. However, the folk solution throws away much information of the original graph. To look for a way to obtain a solution without throwing away as much information as the folk solution, Trudeau~\cite{trudeau2012new} proposed the cycle-complete solution. They made a less extreme transformation of the cost matrix. Bergantiños and Vidal-Puga~\cite{bergantinos2007optimistic} proposed the optimistic game-based solution, where the value of a coalition is defined as the minimal cost of connecting all agents of coalition to the source given that agents outside of coalition are already connected to the source.

The work mentioned above did not consider the situation in which agents strategically block their connections to reduce their cost share and contractors misreport the costs of connections to increase their utilities. Hence, these solutions are insufficient for solving our problem.  

\section{The Model}
We formulate the cost sharing scenario discussed in Section~\ref{introo} using a set of weighted Directed Acyclic Graphs (DAGs) denoted as $\{G_k|G_k=(V \cup \{s\},E,W_k)\}$, where $s$ is the source to which all the other nodes want to connect, $V$ is the set of connected nodes representing agents and $k \in N=\{1,\cdots,n\}$ is a contractor. The directed edge $(i,j) \in E$ represents the connection from $i$ to $j$ and $i$ is called $j$'s parent. The weight $c_{(i,j)}^k \in W_k$ represents the cost of the edge $(i,j) \in E$ incurred by a contractor $k \in N$. A node $i \in V$ connects to the source if there exists a directed path from the source to $i$. The total cost of the connectivity has to be shared among all connected nodes except for $s$. 

Given a set of weighted DAGs, the questions here are how to select a contractor to be the winner, how to determine the payment of the winner of contractors and how to allocate this payment among the nodes. Two strategic behaviors considered are that each node except for the source blocks connections by cutting its outgoing edges and that each contractor misreports the weight of each edge. A directed edge $(i,j)$ cannot be used for connectivity if $i$ cuts it. Our goal is to design a cost sharing mechanism to incentivize contractors to truthfully report all the weights, and to incentivize nodes to offer all their outgoing edges so that we can use all the edges to minimize the total cost of the connectivity.      

The set of $i$'s outgoing edges is called $i$'s {\emph type}, denoted by $r_i$ $(i \in V \cup \{s\})$. The vector of all weights of the contractor $k \in N$ is called $k$'s {\emph type}, denoted by $c_k=(c_{(i,j)}^{k})_{(i,j)\in E}$. The vector of all nodes' types and all contractors' types is called type profile, denoted by $\theta=(r_1,\cdots,r_{|V|+1},c_1,\cdots,c_n)$. The set of all type profiles is called type profile space, denoted by $\Theta$.

We design a cost sharing mechanism that asks each node and each contractor to report their types. Let $r_i' \subseteq r_i$ be the {\emph report} of the node $i$ and $c_k' \in \mathbb{R}^{|\bigcup_{i}r_i'|}$ be the {\emph report} of the contractor $k$, and let $\theta'=(r_1',\cdots,r_{|V|+1}',c_1',\cdots,c_n')$ be the report profile of all nodes and contractors. Finally, given $k \in N$, the DAG induced by $c_k'$ and by $(r_1',\cdots,r_{|V|+1}')$ is denoted by $G_k'=(V \cup \{s\},E',W_k')$, where $E'=\bigcup_{i}r_i' \subseteq E$ and $W_k'$ is the set of reported weights of contractor $k$.

\begin{definition}
A cost sharing mechanism consists of a contractor selection policy $h: \Theta \rightarrow N$, a payment policy of contractors $p: \Theta \rightarrow \mathbb{R}^{n}$, an edge selection policy $f: \Theta \rightarrow 2^{E}$ and a cost sharing policy $x: \Theta \rightarrow \mathbb{R}^{|V|}$. Given a report profile $\theta' \in \Theta$, $h(\theta')\in N$ selects a contractor to be the winner, $p(\theta')=(p_i(\theta'))_{i \in N}$ where $p_i(\theta') \leq 0$ means that the mechanism pays $|p_i(\theta')|$ to the contractor $i$ ($|p_i(\theta')|=0$ if $i \ne h(\theta')$), $f(\theta') \subseteq E$ is the set of the edges connecting the nodes in $V$, and $x(\theta')=(x_i(\theta'))_{i \in V}$ determines the cost share $x_i(\theta')$ of each $i$.
\end{definition}

For simplicity, we use $(h,p,f,x)$ to denote a cost sharing mechanism. Given a report profile $\theta' \in \Theta$, the utility of a contractor $k \in N$ under $(h,p,f,x)$ is defined as 
\begin{equation*}
    u_k(\theta')=
    \begin{cases}
    |p_k(\theta')|-\sum_{(i,j) \in f(\theta')}{c}_{(i,j)}^{k}& \text{$k = h(\theta')$,}\\
    0& \text{$k \ne h(\theta')$,}
    \end{cases}
\end{equation*}
where ${c}_{(i,j)}^{k}$ is the weight of the edge $(i,j)$ incurred by $k$ and $p_k(\theta')$ is the payment of $k$ under $\theta'$.

In the following, we introduce the desirable properties of a cost sharing mechanism.

Truthfulness states that each node is incentivized to offer all its outgoing edges and that each contractor is incentivized to truthfully report the weights of all edges. Note that the source does not behave strategically in this setting.

\begin{definition}
A cost sharing mechanism $(h,p,f,x)$ satisfies \textbf{truthfulness} if $$x_i(r_i,\theta_{-i}') \leq x_i(r_i',\theta_{-i}')$$ and $$u_j(c_j,\theta_{-j}') \geq u_j(c_j', \theta_{-j}'),$$ for all $i \in V$ and for all $j \in N$, where $$\theta_{-i}'=(r_1',\cdots,r_{i-1}',r_{i+1}'.\cdots,r_{|V|+1}',c_1',\cdots,c_n')$$ and $$\theta_{-j}'=(r_1',\cdots,r_{|V|+1}',c_1',\cdots,c_{j-1}',c_{j+1}',\cdots,c_n').$$ 
\end{definition}

Another property is individual rationality, which means that the utility of each contractor is always non-negative.

\begin{definition}
A cost sharing mechanism $(h,p,f,x)$ satisfies \textbf{individual rationality (IR)} if $u_i(\theta')\geq 0$ for all $i \in N$ and for all $\theta' \in \Theta$.  
\end{definition}

We also require that the sum of all nodes' cost share equals the total amount paid to the contractors for all report profiles. That is, the mechanism has no profit or loss. 
\begin{definition}
A cost sharing mechanism $(h,p,f,x)$ satisfies \textbf{budget balance (BB)} if $\sum_{i \in V}x_i(\theta') = \sum_{i \in N}|p_i(\theta')|$ for all $\theta' \in \Theta$.     
\end{definition}

The ranking property requires that for any nodes $i$ and $j$ that have the same parents, if the cost of the edge $(k,i)$ is less expensive than the edge $(k,j)$ for any parent $k$, then $i$ pays less than $j$.
   
\begin{definition}
A cost sharing mechanism $(h,p,f,x)$ satisfies \textbf{ranking} if for all $\theta' \in \Theta$, for all $m \in N$, and for all $i, j \in V$ with $a_i(\theta') \backslash \{j\}=a_j(\theta') \backslash \{i\}$ where $a_i(\theta')$ and $a_j(\theta')$ are respectively the sets of parents of $i$ and $j$, it holds that $c_{(k,i)}^m<c_{(k,j)}^m$ implies $x_i(\theta') < x_j(\theta')$ for all $k \in a_i(\theta') \backslash \{j\}$.
\end{definition}

Positiveness states that each node's cost share should be non-negative. 
\begin{definition}
A cost sharing mechanism $(h,p,f,x)$ satisfies \textbf{positiveness} if $x_i(\theta') \geq 0$ for all $i \in V$ and for all $\theta' \in \Theta$.  
\end{definition}

Finally, we introduce the definition of symmetry which says nodes that play the same role pay the same amount.   
\begin{definition}
A cost sharing mechanism $(h,p,f,x)$ satisfies \textbf{symmetry} if for all $\theta' \in \Theta$, for all $m \in N$, and for all $i, j \in V$ with $a_i(\theta')\backslash \{j\}=a_j(\theta') \backslash \{i\}$ where $a_i(\theta')$ and $a_j(\theta')$ are respectively the sets of parents of $i$ and $j$, it holds that $c_{(k,i)}^m=c_{(k,j)}^m$ implies $x_i(\theta')=x_j(\theta')$ for all $k \in a_i(\theta') \backslash \{j\}$.
\end{definition}
  
In Section~\ref{spm}, we design a cost sharing mechanism to satisfy the above properties.

\section{The Benchmark Mechanisms Based on Shapley Value and Bird Rule on DAGs}
\label{exam}
Our goal is to design a cost sharing mechanism to satisfy truthfulness, budget balance and other properties. To guarantee that contractors truthfully report their types, we select a contractor that corresponds to a minimum spanning tree having the minimal cost, and pay it the second minimal cost.

If we allocate the payment using the DAG with the second minimum cost, nodes have the opportunity to reduce their cost share by strategically cutting their outgoing edges. Therefore, to guarantee that nodes truthfully report their types, we must proportionally allocate the payment using another DAG. The crucial challenge lies in determining the proportion of each node.

The Shapley value is a unique solution enabling the total cost to be shared completely and enabling each node's cost share to depend only on its marginal contribution~\cite{shapley1953value}. Since a node's decision to cut its outgoing edges affects its marginal contribution, we use Shapley value to compute each node's proportion and check whether the Shapley value based mechanism satisfies truthfulness. Another widely used method to compute each node's cost share is the Bird rule, where each node's cost share equals the cost of its incident edge on the MST~\cite{bird1976cost}. Since a node's decision to cut its outgoing edges affects its cost share under the Bird rule, we also use the Bird rule to compute each node's proportion and check whether the Bird rule based mechanism satisfies truthfulness.

\subsection{The Shapley Value Based Mechanism}
\label{svbm}
The key ideas of the mechanism are as follows. Given a report profile, when the contractor is selected and its payment is determined, we use Shapley value to compute each node's proportion. Specifically, we first compute the minimum cost of connecting any subset of nodes to the source. Then we get the marginal cost of adding each node to a subset of other nodes. Furthermore, we get the average marginal cost (Shapley value) of each node. Finally, each node's proportion equals the ratio of its Shapley value to the sum of all nodes' Shapley values. 

We introduce additional notations as follows. Let $j^*$ be the selected contractor, and let $c^{mst}_{j^*}$ denote the cost of $j^*$'s MST, i.e. $$c^{mst}_{j^*}=min_{j \in N}c^{mst}_j.$$ Let $k^*$ be the contractor the cost of whose MST is the second minimum, and let the cost of its MST be denoted by $$c^{mst}_{k^*}=min_{j \neq j^*, j \in N}c^{mst}_j.$$ 

The mechanism is formally described in Algorithm~\ref{shapleymechanism}. Line 1 to line 9 select the winner of contractors and determine the payments of all the contractors. Line 10 determines the set of selected edges. Line 11 to line 17 compute each node's proportion and its cost share.

\begin{algorithm}[tb]
    \caption{The Shapley Value Based Mechanism}
    \label{shapleymechanism}
    \textbf{Input}: A report profile $\theta' \in \Theta$\\
    
    \textbf{Output}: The winner $h(\theta')$,
    the payments $p(\theta')$,

    the edges $f(\theta')$,
    the cost shares $x(\theta')$
    
    \begin{algorithmic}[1] 
        \FOR{each contractor $j \in N$}
        \STATE Compute the MST of $G_j'$ using Prim's algorithm~\cite{prim1957shortest} and compute its cost $c_j^{mst}$;
        \ENDFOR
        \STATE Set $j^*=argmin_{j \in N}c_j^{mst}$;
        \STATE Set $k^*=argmin_{j \neq j^*, j \in N}c_j^{mst}$;
        \STATE Set $h(\theta')=j^*$ and $p_{j^*}(\theta')=-c_{k^*}^{mst}$;
        \FOR{each contractor $j \in N \backslash \{h(\theta')\}$}
        \STATE Set $p_j(\theta')=0$; 
        \ENDFOR
        \STATE Set $f(\theta')$ to be the set of all edges of the MST of $G_{j^*}'$;
        \FOR{each subset $S \subseteq V$}
        \STATE Compute $v(S)$, which is the minimum cost to connect $S$ to $s$;
        \ENDFOR  
        \FOR{each $i \in V$}
        \STATE Compute the Shapley value of $i$, i.e. 
        \begin{equation}
        \label{xiapu}
\phi_i=\sum_{S \subseteq V \backslash \{i\}}\frac{|S|!(|V|-|S|-1)!}{|V|!}\cdot (v(S \cup \{i\})-v(S));    
        \end{equation}
        \STATE Set $x_i(\theta')=\frac{\phi_i}{c_{j^*}^{mst}}\cdot c_{k^*}^{mst}$;
        \ENDFOR
        \STATE \textbf{return} $h(\theta')$, $p(\theta')$, $f(\theta')$, $x(\theta')$
        \end{algorithmic}
\end{algorithm}

Next, we show the properties of Shapley value based mechanism.

\begin{theorem}
The Shapley value based mechanism does not satisfy truthfulness.
\end{theorem}

\begin{figure}[htb]
    \centering
    \includegraphics[width=8.5cm]{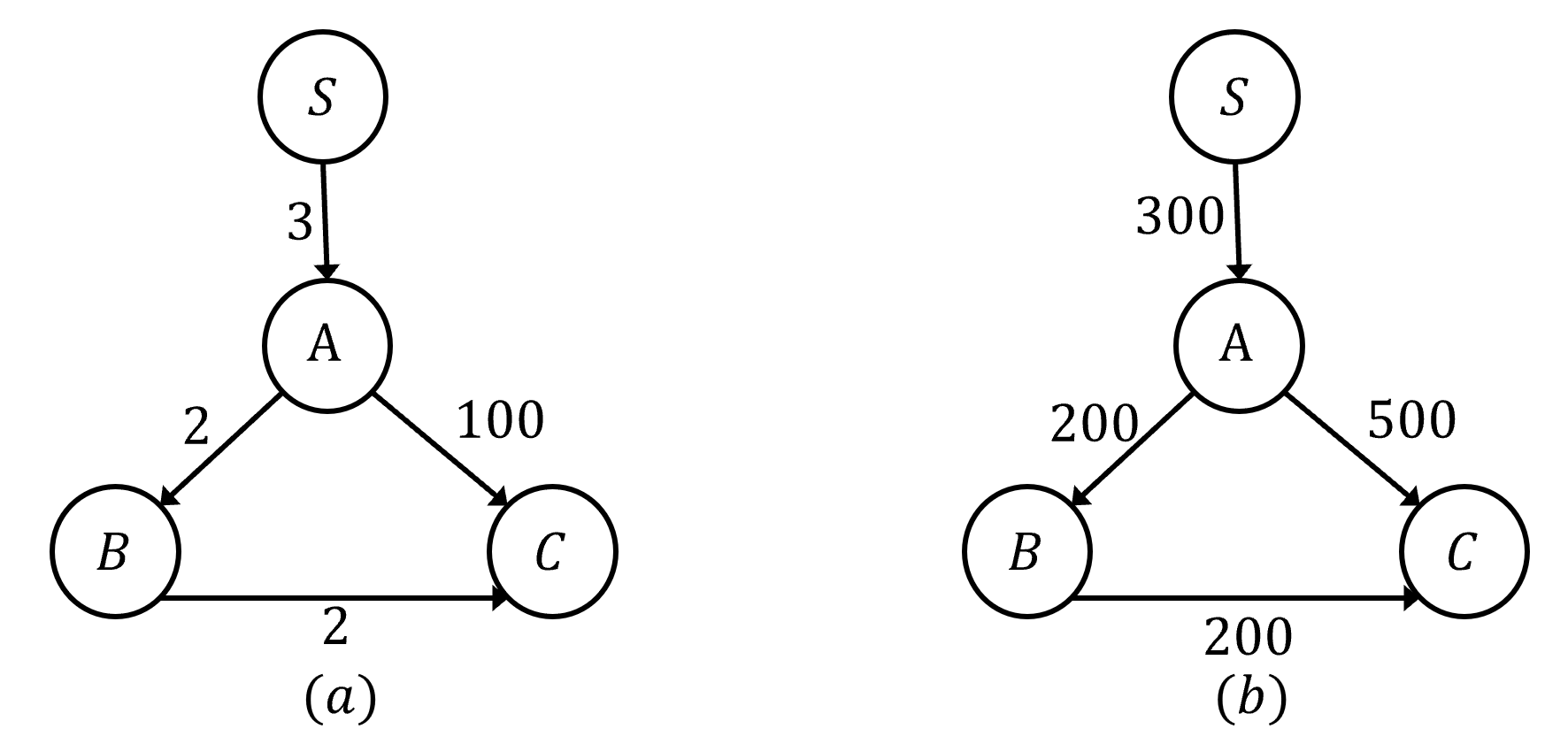}
    \caption{The $s$ represents the source, $A,B,C$ represent the nodes, and the numbers on the edges represent the costs of connections.}
    \label{SV1}
\end{figure}
\begin{figure}[htb]
    \centering
    \includegraphics[width=6.5cm]{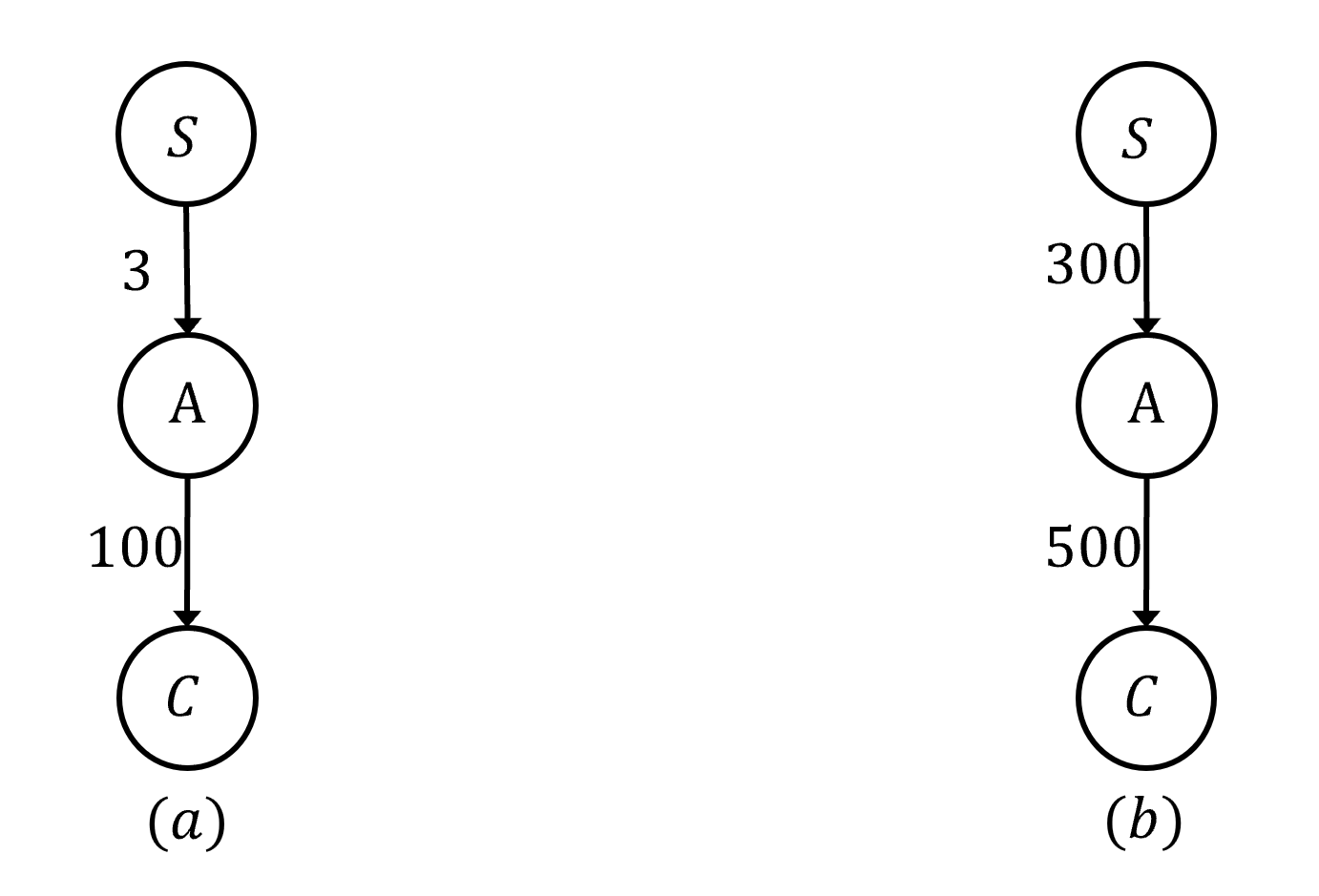}
    \caption{The $s$ represents the source, $A,C$ represent the nodes, and the numbers on the edges represent the costs of connections.}
    \label{SV2}
\end{figure}

\begin{proof}
Since truthfulness is defined for any set of DAGs, it suffices to give a specific set of DAGs where exists a node that is not willing to offer all its outgoing edges. Consider the two DAGs shown in Figure~\ref{SV1}, where the weights are reported by the contractors $a$ and $b$. The costs of minimum spanning trees of $a$ and $b$ are 7 and 700 respectively. So the winner is $a$ and the mechanism pays $a$ 700. Next, we compute the node $A$'s cost share using the left DAG in Figure~\ref{SV1}. First, the coalition values of all subsets of nodes are as follows: 
\begin{equation*}
    \begin{split}
    v(\{B\})&=5, v(\{A,B\})=5, \\
      v(\{A\})&=3, v(\{A,C\})=7, v(\emptyset)=0, \\
       v(\{C\})&=7, v(\{B,C\})=7, v(\{A,B,C\})=7. \\
    \end{split}
\end{equation*}
Second, according to Equation~(\ref{xiapu}), the Shapley values of $A,B$ and $C$ are $1,2$ and $4$ respectively. Finally, the cost share of $A$ is $\frac{1}{7}\cdot 700=100$. 

If the node $A$ cuts the edge $(A,B)$, then $B$ cannot join in the cost sharing since there does not exist a directed path from the source to $B$. The corresponding DAGs are illustrated in Figure~\ref{SV2}. The costs of minimum spanning trees of $a$ and $b$ are 103 and 800 respectively. Thus, the winner is still $a$ and the mechanism pays $a$ 800. Next, we compute the node $A$'s cost share using the left DAG in Figure~\ref{SV2}. The coalition values of all subsets of nodes are as follows: 
\begin{equation*}
    \begin{split}
    v(\{C\})&=103, v(\emptyset)=0, \\ 
    v(\{A\})&=3, v(\{A,C\})=103. \\   
    \end{split}
\end{equation*}
By Equation~(\ref{xiapu}), the Shapley values of $A$ and $C$ are $\frac{3}{2}$ and $\frac{203}{2}$ respectively. Finally, the cost share of $A$ is $$\frac{1.5}{103}\cdot 800 < 100.$$So the node $A$ is not willing to offer all its outgoing edges, i.e. the mechanism does not satisfy truthfulness. 
\end{proof}

\begin{theorem}
The Shapley value based mechanism satisfies budget balance. 
\end{theorem}
\begin{proof}
From Algorithm~\ref{shapleymechanism}, the cost share of each node equals the product of its proportion multiplied by the cost of the second minimum spanning tree where the proportion is the ratio of its Shapley value to the sum of all nodes' Shapley values. Therefore, the sum of all nodes' cost share equals the cost of the second minimum spanning tree. Since the mechanism pays the winner of contractors the cost of the second minimum spanning tree, the mechanism has no profit or loss. Hence, the mechanism satisfies budget balance.
\end{proof}

\begin{theorem}
The Shapley value based mechanism satisfies positiveness. 
\end{theorem}
\begin{proof}
For each node $i \in V$, by Algorithm~\ref{shapleymechanism} and the definition of value function, we get that $v(S \cup \{i\})-v(S)\geq 0$, $\forall S \subseteq V$. Therefore, each node's Shapley value is non-negative. Thus, the proportion of $i$ is non-negative. Since the cost of second minimum spanning tree is positive, each node's cost share is non-negative.      
\end{proof}

\begin{theorem}
\label{svir}
The Shapley value based mechanism satisfies individual rationality.  
\end{theorem}
\begin{proof}
From Algorithm~\ref{shapleymechanism}, if a contractor is not the winner, then its utility is 0. If a contractor is the winner, then the cost of its minimum spanning tree is minimized. Since the mechanism pays the winner the cost of second minimum spanning tree, its utility (the difference between the second minimum of costs of MSTs and the minimum of costs of MSTs) is non-negative. Therefore, each contractor's utility is non-negative, i.e. the mechanism satisfies individual rationality.  
\end{proof}

\subsection{The Bird Rule Based Mechanism}
In Section~\ref{svbm}, we show that the Shapley value based mechanism does not satisfy truthfulness. In this subsection, we check whether the mechanism based on the Bird rule~\cite{bird1976cost} satisfies truthfulness.

The difference between the Bird rule based mechanism and Shapley value based mechanism only lies in the method of computing the proportion of each node. Under the Bird rule based mechanism, we first compute the MST of a DAG by Prim's algorithm. Second, we set the proportion of each node equal to the ratio of the cost of its incident edge on the MST to the cost of the MST.

We restate the following notations. Let $j^*$ be the selected contractor, and let $c^{mst}_{j^*}$ denote the cost of $j^*$'s MST, i.e. $$c^{mst}_{j^*}=min_{j \in N}c^{mst}_j.$$ Let $k^*$ be the contractor the cost of whose MST is the second minimum, and let the cost of its MST be denoted by $$c^{mst}_{k^*}=min_{j \neq j^*, j \in N}c^{mst}_j.$$ 

The mechanism is formally described in Algorithm~\ref{birdmechanism}. Line 1 to line 9 compute the winner of contractors and all the contractors' payments. Line 10 determines the set of selected edges. Line 11 to line 13 compute each node's cost share.    

\begin{algorithm}[tb]
    \caption{The Bird Rule Based Mechanism}
    \label{birdmechanism}
    \textbf{Input}: A report profile $\theta' \in \Theta$\\
    
    \textbf{Output}: The winner $h(\theta')$,
    the payments $p(\theta')$,

    the edges $f(\theta')$,
    the cost shares $x(\theta')$
    
    \begin{algorithmic}[1] 
        \FOR{each contractor $j \in N$}
        \STATE Compute the MST of $G_j'$ using Prim's algorithm and compute its cost $c_j^{mst}$;
        \ENDFOR
        \STATE Set $j^*=argmin_{j \in N}c_j^{mst}$;
        \STATE Set $k^*=argmin_{j \neq j^*, j \in N}c_j^{mst}$;
        \STATE Set $h(\theta')=j^*$ and $p_{j^*}(\theta')=-c_{k^*}^{mst}$;
        \FOR{each contractor $j \in N \backslash \{h(\theta')\}$}
        \STATE Set $p_j(\theta')=0$; 
        \ENDFOR
        \STATE Set $f(\theta')$ to be the set of all edges of the MST of $G_{j^*}'$;
        \FOR{each $i \in V$}
        \STATE Set $x_i(\theta')=\frac{k_i}{c_{j^*}^{mst}}\cdot c_{k^*}^{mst}$, where $k_i$ is the cost of $i$'s incident edge on the MST of $G_{j^*}'$;
        \ENDFOR
        \STATE \textbf{return} $h(\theta')$, $p(\theta')$, $f(\theta')$, $x(\theta')$
        \end{algorithmic}
\end{algorithm}

Next, we show the properties of Bird rule based mechanism.

\begin{theorem}
The Bird rule based mechanism does not satisfy truthfulness.
\end{theorem}

\begin{figure}[htb]
    \centering
    \includegraphics[width=9.2cm]{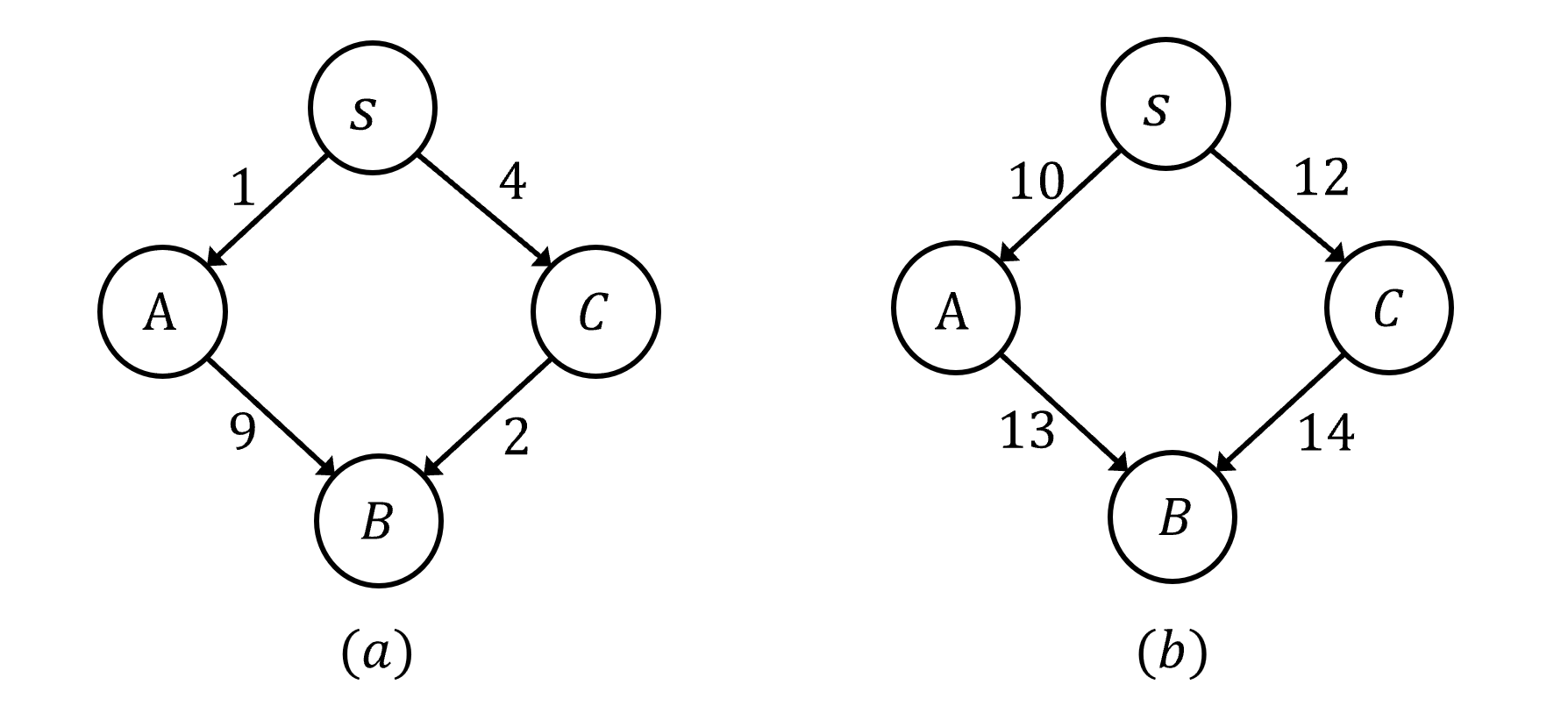}
    \caption{The $s$ represents the source, $A,B,C$ represent the nodes, and the numbers on the edges represent the costs of connections.}
    \label{BR1}
\end{figure}
\begin{figure}[htb]
    \centering
    \includegraphics[width=8.5cm]{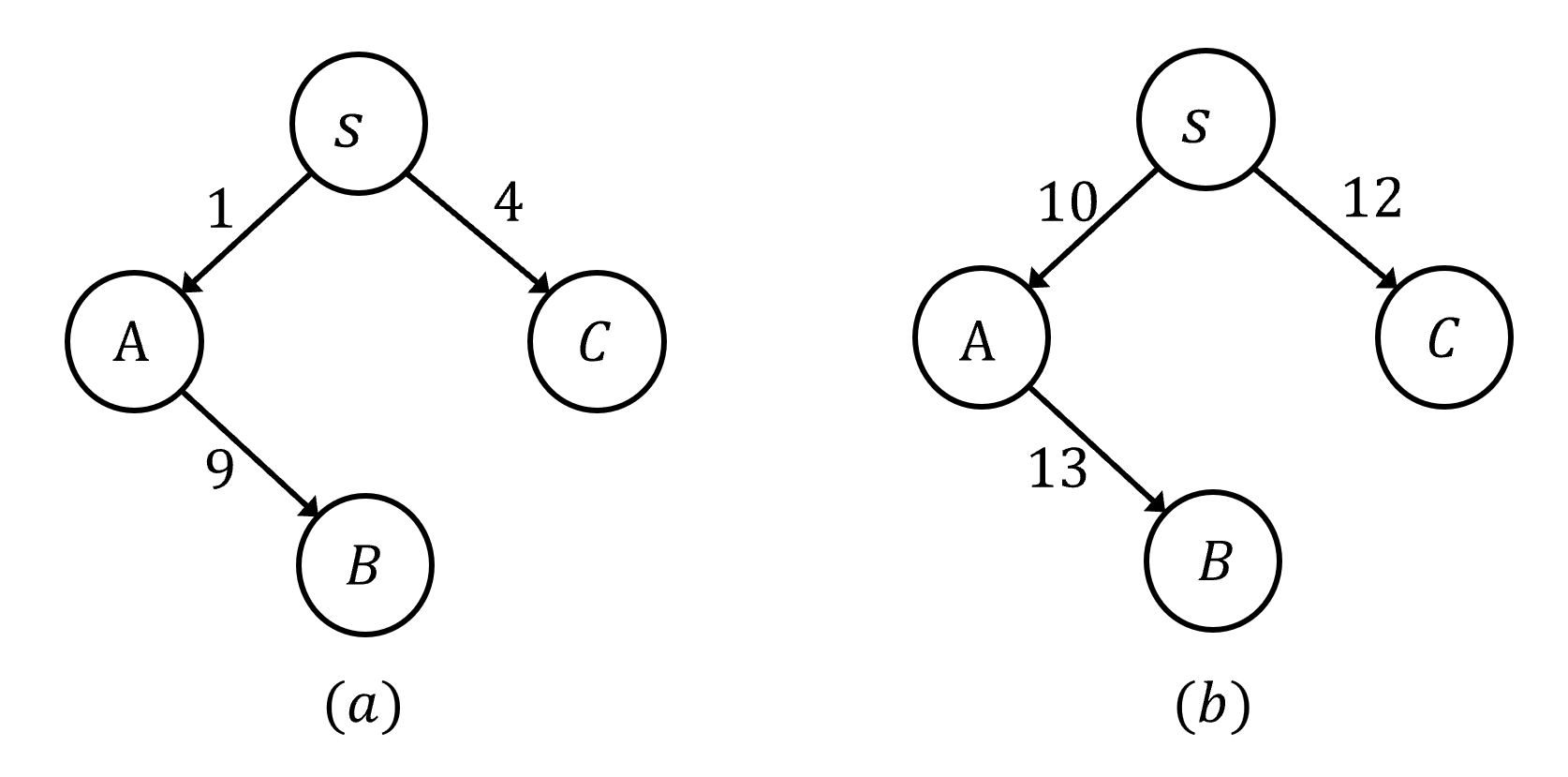}
    \caption{The $s$ represents the source, $A,B,C$ represent the nodes, and the numbers on the edges represent the costs of connections.}
    \label{BR2}
\end{figure}
\begin{proof}
Since truthfulness is defined for any set of DAGs, it suffices to give a specific set of DAGs where exists a node that is not willing to offer all its outgoing edges. Consider the two DAGs shown in Figure~\ref{BR1}, where the weights are reported by the contractors $a$ and $b$. The costs of minimum spanning trees of $a$ and $b$ are 7 and 35 respectively. So the winner is $a$ and the mechanism pays $a$ 35. Next, we compute the node $C$'s cost share on the left DAG in Figure~\ref{BR1}. The cost of its incident edge on the MST of the left DAG is 4. Thus the cost share of $C$ is $\frac{4}{7} \cdot 35$.       

If the node $C$ cuts the edge $(C,B)$, the corresponding DAGs are illustrated in Figure~\ref{BR2}. The costs of minimum spanning trees of $a$ and $b$ are 14 and 35 respectively. So the winner is still $a$ and the mechanism pays $a$ 35. Next, we compute the node $C$'s cost share on the left DAG in Figure~\ref{BR2}. The cost of its incident edge on the MST of the left DAG is still 4. Thus the cost share of $C$ is $$\frac{4}{14}\cdot 35 < \frac{4}{7} \cdot 35.$$So the node $C$ is not willing to offer all its outgoing edges, i.e. the mechanism does not satisfy truthfulness.
\end{proof}

\begin{theorem}
The Bird rule based mechanism satisfies budget balance. 
\end{theorem}
\begin{proof}
From Algorithm~\ref{birdmechanism}, the cost share of each node equals the product of its proportion multiplied by the cost of the second minimum spanning tree. The proportion is the ratio of the cost of its incident edge on the MST to the cost of MST. Therefore, the sum of all nodes' cost share equals the cost of the second minimum spanning tree. Since the mechanism pays the winner of contractors the cost of the second minimum spanning tree, the mechanism has no profit or loss. Hence, the mechanism satisfies budget balance.
\end{proof}

\begin{theorem}
The Bird rule based mechanism satisfies positiveness.
\end{theorem}

\begin{proof}
The cost of each node's incident edge on the MST is non-negative. Then the proportion of each node is non-negative. Since the cost of second minimum spanning tree is positive, each node's cost share is non-negative. That is, the mechanism satisfies positiveness.       
\end{proof}

\begin{theorem}
The Bird rule based mechanism satisfies individual rationality.  
\end{theorem}

\begin{proof}
From Algorithm~\ref{birdmechanism}, if a contractor is not the winner, then its utility is 0. If a contractor is the winner, then the cost of its minimum spanning tree is minimized. Since the mechanism pays the winner the cost of second minimum spanning tree, its utility (the difference between the second minimum of costs of MSTs and the minimum of costs of MSTs) is non-negative. Therefore, each contractor's utility is non-negative, i.e. the mechanism satisfies individual rationality.   
\end{proof}

In summary, although both the Shapley value based mechanism and Bird rule based mechanism satisfy budget balance, positiveness and individual rationality, they do not satisfy truthfulness. Therefore, in the next section, we propose a novel cost sharing mechanism that satisfies truthfulness and other properties.

\section{The Shortest Path Based Mechanism}
\label{spm}
The reason why the mechanisms based on Shapley value and Bird rule do not satisfy truthfulness is that nodes can cut their outgoing edges to reduce their proportions and thus their cost share decreases. Therefore, to satisfy truthfulness, we need to guarantee that cutting outgoing edges does not decrease each node's proportion.

In this section, we propose a novel cost sharing mechanism based on the shortest path that satisfies truthfulness. In addition, we show that it also satisfies budget balance, positiveness, individual rationality, ranking and symmetry. 

The key ideas of the proposed mechanism are as follows. To incentivize the contractors to truthfully report their types, the mechanism selects a contractor to be the winner whose minimum spanning tree has the minimum cost among minimum spanning trees of all DAGs. The winner's payment is defined as the negative second minimum cost. To incentivize the nodes to truthfully report their types, we first select the DAG with the minimum cost of minimum spanning tree to compute each node's proportion. Second, we use the Breadth-First-Search (BFS) algorithm~\cite{silvela2001breadth} to compute each node's depth and sort all nodes by their depths in descending order (with random tie-breaking). Third, following the above order, we compute the shortest path from the source to each node. Fourth, we update the costs on this shortest path to zero. Fifth, we compute each node's proportion, which is defined as the ratio of the cost of its shortest path to the sum of costs of all nodes' shortest paths. Finally, we select the edges with zero cost. 

\begin{definition}
The length of a directed path is the number of edges on it. For each node $i \in V$, its depth is defined as the length of the shortest directed path from the source to it. 
\end{definition}

Now we introduce additional notations. For each node $i$, let $S_i$ denote the set consisting of the source $s$, the nodes with the depths larger than $i$ and the nodes on their shortest paths. Let $d(i,S_i)$ be the minimum of costs from each node in $S_i$ to $i$, $E(i,S_i)$ be the set of all edges of the directed path corresponding to $d(i,S_i)$, and $V(i,S_i)$ be the set of all nodes of the directed path corresponding to $d(i,S_i)$. Let $j^*$ be the selected contractor, $c^{mst}_{j^*}$ be the cost of $j^*$'s MST, $k^*$ be the contractor the cost of whose MST is the second minimum, and $c^{mst}_{k^*}$ be the cost of $k^*$'s MST.

The proposed mechanism is formally described in Algorithm~\ref{proposedmechanism}. Line 1 to line 10 determine the winner of contractors and the payments of all contractors. Line 11 to line 21 compute each node's proportion and its cost share. Line 22 defines the set of selected edges.  

A running example is given in Example~\ref{proposedexample}. 

\begin{algorithm}[tb]
    \caption{The Shortest Path Based Mechanism}
    \label{proposedmechanism}
    \textbf{Input}: A report profile $\theta' \in \Theta$\\
    
    \textbf{Output}: The winner $h(\theta')$,
    the payments $p(\theta')$,

    the edges $f(\theta')$,
    the cost shares $x(\theta')$
    
    \begin{algorithmic}[1] 
        \STATE Initialize $A=\emptyset$ and $B=0$;
        \FOR{each contractor $j \in N$}
        \STATE Compute the MST of $G_j'$ using Prim's algorithm and compute its cost $c_j^{mst}$;
        \ENDFOR
        \STATE Set $j^*=argmin_{j \in N}c_j^{mst}$;
        \STATE Set $k^*=argmin_{j \neq j^*, j \in N}c_j^{mst}$;
        \STATE Set $h(\theta')=j^*$ and $p_{j^*}(\theta')=-c_{k^*}^{mst}$;
        \FOR{each contractor $j \in N \backslash \{h(\theta')\}$}
        \STATE Set $p_j(\theta')=0$; 
        \ENDFOR
        \STATE Sort nodes in descending order by their depths on $G_{j^*}'$ (assuming the sequence is $1,\cdots,|V|$);
        \STATE Set $S_1=\emptyset$;
        \FOR{each node $i=1,\cdots,|V|$}
        \STATE Compute $d(i,S_i)$, $E(i,S_i)$ and $V(i,S_i)$;
        \STATE Set $S_{i+1}=S_i \cup V(i,S_i)$;
        \STATE Set $A=A \cup E(i,S_i)$;
        \STATE Set $B=B+d(i,S_i)$;
        \ENDFOR  
        \FOR{each node $i=1,\cdots,|V|$}
        \STATE Set $x_i(\theta')=\frac{d(i,S_i)}{B}\cdot c_{k^*}^{mst}$;
        \ENDFOR
        \STATE Set $f(\theta')=A$;
        \STATE \textbf{return} $h(\theta')$, $p(\theta')$, $f(\theta')$, $x(\theta')$
        \end{algorithmic}
\end{algorithm}
\begin{figure}[htb]
    \centering
    \includegraphics[width=8.5cm]{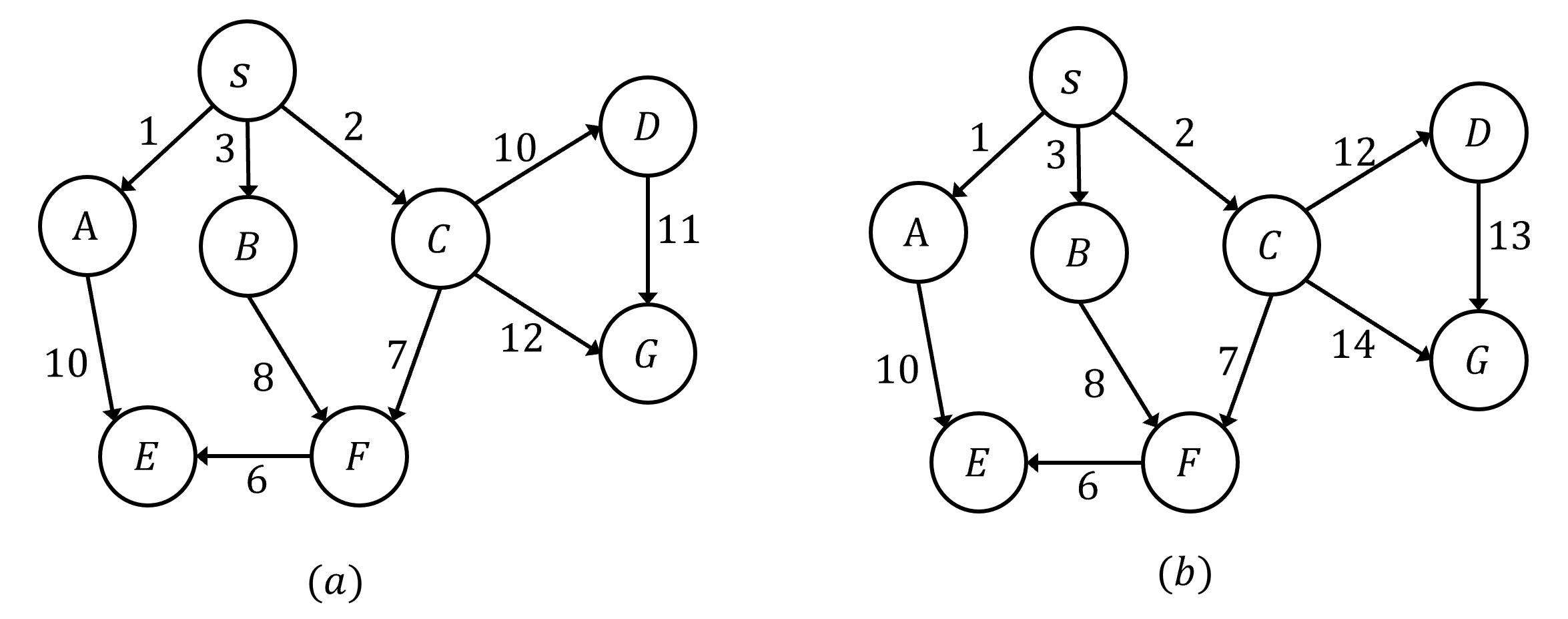}
    \caption{The $s$ represents the source, the letters in circles represent the nodes, and the numbers on the edges represent the costs of connections.}
    \label{pro0}
\end{figure}

\begin{figure}[htb]
    \centering
    \includegraphics[width=8.5cm]{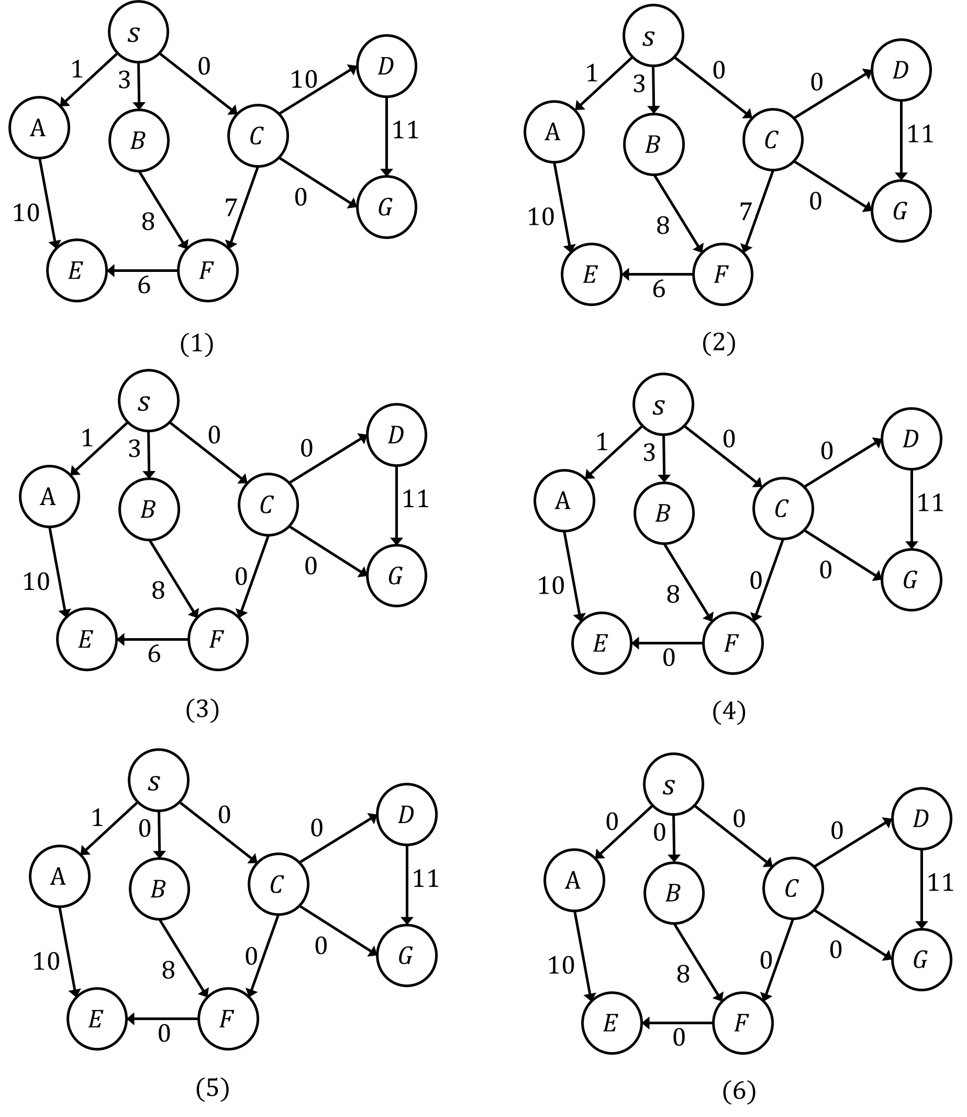}
    \caption{The process of updating the costs of edges}
    \label{pro4}
\end{figure}

\begin{example}
\label{proposedexample}
The DAGs generated by a report profile $\theta' \in \Theta$ are illustrated in Figure~\ref{pro0}, where $N=\{a,b\}$ and $V=\{A,B,C,D,E,F,G\}$. According to each contractor's report, we obtain that the costs of minimum spanning trees of $a$ and $b$ are 40 and 44 respectively. Therefore, we have $h(\theta')=a$, $p_a(\theta')=-44$ and $p_b(\theta')=0$. Next, we compute the cost shares of all nodes on the left DAG in Figure~\ref{pro0}. The nodes are ordered by their depths: $G,D,F,E,C,B,A$. For node $G$, we get $S_G=\{s\}$. The shortest directed path from the nodes in $S_G$ to $G$ is $s\rightarrow C \rightarrow G$. Thus we have $d(G,S_G)=2+12=14$. We set the costs of the edges $(s,C)$ and $(C,G)$ to zero. For node $D$, we get $S_D=\{s,C,G\}$. The shortest directed path from the nodes in $S_D$ to $D$ is $C \rightarrow D$. So we have $d(D,S_D)=10$. The cost of the edge $(C,D)$ is set to zero. Similarly, we obtain 
\begin{equation*}
    \begin{split}
    d(F,S_F)&=7, d(E,S_E)=6, \\
    d(C,S_C)&=0, d(B,S_B)=3, d(A,S_A)=1. \\    
    \end{split}
\end{equation*}
So, by the equation in Algorithm~\ref{proposedmechanism} $$x_i(\theta')=\frac{d(i,S_i)}{B}\cdot c_{k^*}^{mst},$$we obtain that the cost shares of nodes $G,D,F,E,C,B,A$ are respectively
\begin{equation*}
    \begin{split}
     x_G(\theta')&=\frac{14}{41}\cdot 44, x_D(\theta')=\frac{10}{41}\cdot 44, x_F(\theta')=\frac{7}{41}\cdot 44, \\
     x_E(\theta')&=\frac{6}{41}\cdot 44, x_C(\theta')=0, x_B(\theta')=\frac{3}{41}\cdot 44, x_A(\theta')=\frac{1}{41}\cdot 44.  \\ 
    \end{split}
\end{equation*}
Finally, we get the set of selected edges $$f(\theta')=\{(s,C),(s,B),(s,A),(C,D),(C,G),(C,F),(F,E)\}.$$The process of updating the costs of edges is illustrated in Figure~\ref{pro4} ((1)-(6)).
\end{example}

\subsection{Properties of the Shortest Path Based Mechanism}
In the following, we show the nice properties of the proposed mechanism. 

\begin{theorem}
The shortest path based mechanism satisfies truthfulness. 
\end{theorem}
\begin{proof}
First, we prove that each contractor $i \in N$ will report all the weights truthfully. To prove this, we only need to show that when $i$ truthfully reports the weights, its utility is larger than or equal to the utility when it misreports the weights. 

When $i$ truthfully reports its weights, there are two cases.

\begin{itemize}
    \item $i = h(\theta')$. According to our mechanism, $i$'s utility is the difference between the second minimum of costs of MSTs and the minimum of costs of MSTs, which is non-negative. 
    
    If $i$ reports larger weights, there are two possibilities.
    \begin{itemize}
        \item $i \ne h(\theta')$. Then $i$'s utility is 0.  
        \item $i = h(\theta')$. Then the second minimum of costs of MSTs does not change and so does $i$'s utility. 
    \end{itemize}
    
If $i$ reports lower weights, then it is still the winner and its utility does not change.   
    \item $i \ne h(\theta')$. Then its utility is 0. 
    
    If $i$ reports larger weights, then it is still not the winner and its utility does not change. 
    
    If $i$ reports lower weights, there are two possibilities.
    \begin{itemize}
        \item $i \ne h(\theta')$. Then its utility is still 0.
        \item $i = h(\theta')$. Then its utility is negative. 
    \end{itemize}
\end{itemize}

So truthful report maximizes its utility.

Second, we prove that each node $i \in V$ will offer all its outgoing edges. 

When $i$ offers all its outgoing edges, its cost share is $$\frac{d(i,S_i)}{\sum_{i \in V}d(i,S_i)}\cdot c_{k^*}^{mst}.$$

If node $i$ does not offer all its outgoing edges, there are two possibilities.
\begin{itemize}
    \item If there exists a node for which $i$ is on the directed path from the source to that node, then the cost share of $i$ is 0. Thus, the cost share of $i$ will weakly increase. 
    \item Otherwise, we have that $d(j,S_j) (j \neq i)$ remains unchanged. By the definition of $d(i,S_i)$, it also does not change. So $i$'s proportion $$\frac{d(i,S_i)}{\sum_{i \in V}d(i,S_i)}$$ does not change. In addition, $c_{k^*}^{mst}$ will weakly increase. Hence, the cost share of $i$ will weakly increase. 
\end{itemize}

So truthful report minimizes its cost share.

Putting the above analysis together, the proposed mechanism satisfies truthfulness.
\end{proof}

\begin{theorem}
The shortest path based mechanism satisfies budget balance.
\end{theorem}

\begin{proof}
We need to prove that given $\theta' \in \Theta$, the sum of all nodes' cost share equals the second minimum of costs of minimum spanning trees, i.e. $$\sum_{i \in V}x_i(\theta')=c_{k^*}^{mst}.$$According to line 17 in Algorithm~\ref{proposedmechanism}, we have $$\sum_{i \in V}d(i,S_i)=B.$$According to line 20 in Algorithm~\ref{proposedmechanism}, we get $$\sum_{i \in V}x_i(\theta')=\frac{\sum_{i \in V}d(i,S_i)}{B}\cdot c_{k^*}^{mst}=c_{k^*}^{mst}.$$ 
\end{proof}

\begin{theorem}
The shortest path based mechanism satisfies individual rationality. 
\end{theorem}

\begin{proof}
According to the proposed mechanism, the utility of a contractor is zero if it is not the winner. For the winner, since the mechanism pays it the second minimum of costs of minimum spanning trees, its utility (the difference between the second minimum of costs of MSTs and the minimum of costs of MSTs) is non-negative. Therefore, the utility of each contractor is non-negative. 
\end{proof}

\begin{theorem}
The shortest path based mechanism satisfies positiveness. 
\end{theorem}

\begin{proof}
According to the proposed mechanism, given $\theta' \in \Theta$, the cost share of each $i \in V$ is computed as $$x_i(\theta')=\frac{d(i,S_i)}{\sum_{i \in V}d(i,S_i)}\cdot c_{k^*}^{mst}.$$Since both $d(i,S_i)$ and $c_{k^*}^{mst}$ are non-negative, the cost share of $i$ is also non-negative. 
\end{proof}

\begin{theorem}
\label{spmran}
The shortest path based mechanism satisfies ranking.
\end{theorem}

\begin{proof}
We need to show that, given $\theta' \in \Theta$, the winner of contractors $m \in N$, and $i,j \in V$ with $a_i(\theta')\backslash \{j\}=a_j(\theta')\backslash \{i\}$, it holds that $c_{(k,i)}^m < c_{(k,j)}^m$ ($\forall k \in a_i(\theta') \backslash \{j\}$) implies $x_i(\theta')<x_j(\theta')$. 

From Algorithm~\ref{proposedmechanism}, we know $$x_i(\theta')=\frac{d(i,S_i)}{\sum_{l \in V}d(l,S_l)}\cdot c_{k^*}^{mst}$$ and $$x_j(\theta')=\frac{d(j,S_j)}{\sum_{l \in V}d(l,S_l)}\cdot c_{k^*}^{mst}.$$We only need to compare $d(i,S_i)$ with $d(j,S_j)$. The $d(i,S_i)$ represents the minimum of costs from each node in $S_i$ to $i$ and the $d(j,S_j)$ represents the minimum of costs from each node in $S_j$ to $j$. Due to the definition of ranking, the depth of $i$ equals the depth of $j$. Therefore, we have $S_i=S_j$. Since $c_{(k,i)}^m < c_{(k,j)}^m$ ($\forall k \in a_i(\theta') \backslash \{j\}$), we have $d(i,S_i)<d(j,S_j)$. Therefore, it holds that $x_i(\theta')<x_j(\theta')$.
\end{proof}

\begin{theorem}
The shortest path based mechanism satisfies symmetry.
\end{theorem}

\begin{proof}
We need to show that, given $\theta' \in \Theta$, the winner of contractors $m \in N$, and $i,j \in V$ with $a_i(\theta')\backslash \{j\}=a_j(\theta')\backslash \{i\}$, it holds that $c_{(k,i)}^m = c_{(k,j)}^m$ ($\forall k \in a_i(\theta') \backslash \{j\}$) implies $x_i(\theta')=x_j(\theta')$. 

From Algorithm~\ref{proposedmechanism}, we know $$x_i(\theta')=\frac{d(i,S_i)}{\sum_{l \in V}d(l,S_l)}\cdot c_{k^*}^{mst}$$ and $$x_j(\theta')=\frac{d(j,S_j)}{\sum_{l \in V}d(l,S_l)}\cdot c_{k^*}^{mst}.$$The $d(i,S_i)$ represents the minimum of costs from each node in $S_i$ to $i$ and the $d(j,S_j)$ represents the minimum of costs from each node in $S_j$ to $j$. Since the depth of $i$ equals the depth of $j$, it follows that $S_i=S_j$. Since $c_{(k,i)}^m = c_{(k,j)}^m$ ($\forall k \in a_i(\theta') \backslash \{j\}$), we have $d(i,S_i)=d(j,S_j)$. Further, we infer that $x_i(\theta')=x_j(\theta')$.
\end{proof}

\section{Conclusions}
For the first time, we study how to design a cost sharing mechanism to prevent the agents' strategic behaviors, i.e. the contractors' misreporting the weights of edges and the nodes' cutting their outgoing edges. We model the two strategic behaviors and define key properties that cost sharing mechanisms satisfy. Then we show that the cost sharing benchmark mechanisms based on Shapley value and Bird rule cannot satisfy truthfulness. Finally, we propose a novel cost sharing mechanism based on shortest path and show it satisfies truthfulness and other properties.

In the future, we try to establish sufficient and necessary conditions for a cost sharing mechanism to satisfy truthfulness.  

\begin{acks}
This work is supported by Science and Technology Commission of Shanghai Municipality (No. 23010503000 and No. 22ZR1442200), and Shanghai Frontiers Science Center of Human-centered Artificial Intelligence (ShangHAI).
\end{acks}

\bibliographystyle{ACM-Reference-Format}
\bibliography{sample-base}

\end{document}